\documentclass[runningheads]{llncs}
\begin{document}

\newtheorem{sentence}{Sentence}
\newtheorem{follow}{Corollary}
\newtheorem{consequence}{Consequence}
\newcommand{\bP}{{\mathbf P}}
\newcommand{\bQ}{{\mathbf Q}}
\newcommand{\bR}{{\mathbf R}}
\newcommand{\bL}{{\mathbf L}}
\newcommand{\bM}{{\mathbf M}}
\newcommand{\bU}{{\mathbf U}}
\newcommand{\bA}{{\mathbf A}}
\newcommand{\bB}{{\mathbf B}}
\newcommand{\bC}{{\mathbf C}}
\newcommand{\bD}{{\mathbf D}}
\newcommand{\Dmtr}[4]%
{\left|\begin{array}{rr}#1 & #2\\#3 & #4\end{array}\right|}
\def\cA{{\bf \mathcal A}}
\def\cC{{\bf \mathcal C}}
\def\cD{{\bf \mathcal D}}
\newcommand{\LDU}{{\mathrm{LDU}}}
\newcommand{\Adj}{{\mathrm{Adj}}}
\newcommand{\diag}{{\mathrm{diag}}}
\newcommand{\rank}{{\mathrm{rank}}}
\def\al{{\alpha}}
\def\be{\beta}
\def\de{\delta}
\def\ve{\varepsilon}
\def\si{\sigma}
\def\ga{\gamma}
\def\di{{\bf {\rm diag}}}
\def\De{\Delta}
\def\G{\Gamma}
\def\cA{{\bf \mathcal A}}
\def\cU{{\bf \mathcal U}}
\def\cL{{\bf \mathcal L}}
\def\cH{{\bf \mathcal H}}
\def\cG{{\bf \mathcal G}}

\title{Generalized Bruhat decomposition in commutative domains}
\titlerunning{Generalized Bruhat decomposition in commutative domains} 
\author{Gennadi Malaschonok
   \thanks{Preprint of the paper: G.I.Malaschonok. Generalized Bruhat decomposition in commutative domains / in book: 
   Computer Algebra in Scientific Computing. CASC'2013. LNCS 8136, Springer, Heidelberg, 2013, pp.231-242.
   Supported by part Russian Foundation for Basic Research No. 12-07-00755a}}
\institute{Tambov State University, \\  Internatsionalnaya 33,
392622 Tambov, Russia \\
\email{malaschonok@gmail.com}}

\maketitle

\begin{abstract}
Deterministic recursive algorithms for the computation of generalized  Bruhat decomposition
of the matrix in commutative domain are presented.
This method has the same complexity as the algorithm of matrix multiplication.
\end{abstract}


\section{Introduction}

A matrix decomposition of a form  $A=VwU$ is called the Bruhat decomposition of the matrix $A$, if
$V$ and $U$ are nonsingular upper triangular matrices and $w$ is a
matrix of permutation. It is usually assumed that the matrix $A$ is defined in a certain field.
 Bruhat decomposition plays  an important
role in the theory of algebraic groups. The generalized Bruhat decomposition was
introduced and developed by D.Grigoriev\cite{10},\cite{11}.  

In \cite{12} there was constructed a pivot-free
matrix decomposition method in a common case of singular matrices
over a field of arbitrary characteristic.  This algorithm has the same complexity as 
matrix multiplication and does not require pivoting. For singular
matrices it allows us to obtain a nonsingular block of the biggest
size.

Let $R$ be a commutative domain, $F$ be the field of fractions over $R$.
We want to obtain a decomposition of matrix $A$ over domain $R$ in the form  $A=VwU$, where $V$ and $U$ are
upper triangular matrices over $R$ and $w$ is a matrix of permutation, which is multiplied by some diagonal matrix in the
field of fractions $F$. Moreover each nonzero element of $w$ has the form $(a^i a^{i-1})^{-1}$, where $a^i$ is some minor of order $i$ of matrix $A$.

We call such triangular decomposition the Bruhat decomposition in the commutative domain $R$.

In \cite{15} a fast algorithm for adjoint matrix computation was proposed.   
On the basis of this algorithm for computing the adjoint matrix in the \cite{17} proposed a fast algorithm for $LDU$ decomposition. However, this algorithm is required to calculate the adjoint matrix and use it to calculate $LDU$ decomposition.

In this paper, we propose another algorithm that does not rely on the calculation of the adjoint matrix 
and which costs less number of operations.
We   construct the decomposition in the form $A=LDU$, where $L$ and $U$ are lower
and upper triangular matrices, $D$  is a matrix of permutation, which is multiplied by some diagonal matrix in the
field of fractions $F$ and  has the same rank as the matrix $A$. Then the Bruhat decomposition $VwU$ in the domain $R$ may be easily obtained using the
matrices $L$, $D$ and $U$.


\section{Triangular decomposition in domain}

Let $R$ be a commutative domain, $A=(a_{i,j}) \in R^{n \times n}$ be a matrix of order $n$, $\alpha^k_{i,j}$ be $k\times k$ minor of matrix $A$ which disposed in the rows $1,2, \ldots ,k-1,i$ and columns   $1,2, \ldots ,k-1,j$ for all integers $  i,j,k \in \{1, \ldots ,n\}$. We suppose that the row $i$ of the matrix $A$ is situated at the last row of the minor, and the column $j$ of the matrix $A$ is situated at the last column of the minor. We denote $\alpha^0=1$ and $\alpha^k = \alpha^k_{k,k}$ for all diagonal minors ($1\leq k \leq n $). And we use the notation $\de_{ij}$ for Kronecker delta.
\smallskip

Let $k$ and $s$ be  integers in the interval $0\leq k<s\leq n$, $\cA_s^k=(\al^{k+1}_{i,j})$ be the matrix of minors with size $(s-k)\times (s-k)$ which has elements $\al^{k+1}_{i,j}$, $i,j=k+1, \ldots ,s-1,s,$   and  $\cA_n^0=(\al^{1}_{i,j})=A$.

We shall use the following identity (see \cite{13}, \cite{14}):
\begin{theorem}[Sylvester determinant identity]
 
\noindent Let $k$ and $s$ be an integers in the interval $0\leq k<s\leq n$. Then it is true that
$$
\det (\cA_s^k)=\al^s (\al^k)^{s-k-1}.
\eqno (1)$$
\end{theorem}

\begin{theorem}[LDU decomposition of the minors matrix]

\noindent Let $A=(a_{i,j}) \in R^{n \times n}$ be the matrix of rank $r$, $\alpha^i\neq 0$ for $i=k,k+1, \ldots ,r$, $r\leq s\leq n$, then the matrix of minors $\cA_s^k$ is equal to the following product of three matrices:
$$
\cA_s^k = L^k_sD^k_sU^k_s  =(a^j_{i,j})  (\de_{ij}\alpha^k(\alpha^{i-1} \alpha^{i} )^{-1}) (a^i_{i,j}).
\eqno(2)
$$
\noindent The matrix $ L^k_s=(a^j_{i,j})$, $i=k+1 \ldots s$, $j=k+1 \ldots r$, is a low triangular matrix of size $(s-k)\times (r-k)$,
 the matrix $U^k_s=(a^i_{i,j})$, $i=k+1 \ldots r$, $j=k+1 \ldots s$, is an upper triangular matrix of size $(r-k)\times (s-k)$ and
$D^k_s= (\de_{ij}\alpha^k(\alpha^{i-1} \alpha^{i} )^{-1})$, $i=k+1 \ldots r$, $j=k+1 \ldots r$, is a diagonal matrix of size $(r-k)\times (r-k)$.
\end{theorem}

\begin{proof}

\noindent
Let us write the matrix equation (2) for $k+1=r$
$$
(a^{k+1}_{i,j})=(a^{k+1}_{i,k+1})(\de_{k+1,k+1}a^{k}(a^{k}a^{k+1})^{-1})(a^{k+1}_{k+1,j})
\eqno(3)
$$
This equation is correct due to Sylvester determinant identity
$$
a^{k+1}_{i,j}a^{k+1} -a^{k+1}_{i,k+1}a^{k+1}_{k+1,j}=a^{k+2}_{i,j} a^{k},
\eqno(4)
$$
and the equality $a^{k+2}_{i,j}=0$. This equality is a consequence of
the fact that minors $a^{k+2}_{i,j}$
have the order greater then the rank of the matrix $A$.

Let for all $h$, $k<h< r $, the statement (1) be correct for matrices $\cA_s^h=(a^{h+1}_{i,j})$.
 We have to prove it for $h=k$.
Let us write one matrix element in (2) for the matrix $\cA_s^{k+1}=(a^{k+2}_{i,j})$ :
$$
a^{k+2}_{i,j}=\sum_{t=k+2}^{min(i,j,r)}
   a^t_{i,t}  \alpha^{k+1}(\alpha^{t-1} \alpha^{t} )^{-1} a^t_{t,j}.
$$
We have to prove the corresponding expression for the elements of the matrix $\cA_s^{k}$. Due to the
Sylvester determinant identity (3)
we obtain
$$a^{k+1}_{i,j}=a^{k+1}_{i,k+1}  (\alpha^{k+1} )^{-1} a^{k+1}_{k+1,j}+\alpha^k (\alpha^{k+1})^{-1}a^{k+2}_{i,j}=$$
$$
a^{k+1}_{i,k+1}  \alpha^k(\alpha^{k} \alpha^{k+1} )^{-1} a^{k+1}_{k+1,j}+\alpha^k (\alpha^{k+1})^{-1}
\sum_{t=k+2}^{min(i,j,r)}
   a^t_{i,t}  \alpha^{k+1}(\alpha^{t-1} \alpha^{t} )^{-1} a^t_{t,j}=$$
$$\sum_{t=k+1}^{min(i,j)}
   a^t_{i,t}  \alpha^{k}(\alpha^{t-1} \alpha^{t} )^{-1} a^t_{t,j}.$$
  
\end{proof}
%

\begin{consequence}[LDU decomposition of matrix $A$]

\par\noindent
Let $A=(a_{i,j}) \in R^{n \times n}$, be the matrix of rank $r$, $r\leq n$, $\alpha^i\neq 0$ for
$i=1,2, \ldots ,r$, then matrix  $A$ is equal to the following product of three matrices:
$$
 A = L^0_n D^0_n U^0_n = (a^j_{i,j})  (\de_{ij} (\alpha^{i-1} \alpha^{i} )^{-1}) (a^i_{i,j}).
\eqno(4)
$$
\noindent The matrix $ L^0_n=(a^j_{i,j})$, $i= 1 \ldots n$, $j= 1 \ldots r$, is a low triangular matrix of size $n\times  r $,
the matrix $U^0_n=(a^i_{i,j})$, $i= 1 \ldots r$, $j= 1 \ldots n$, is an upper triangular matrix of size $r \times n$ and
$D^0_n= (\de_{ij} (\alpha^{i-1} \alpha^{i} )^{-1})$, $i= 1 \ldots r$, $j= 1 \ldots r$, is a diagonal matrix of size $ r \times r $.
\end{consequence}

Let $I_n$ be the identity matrix and $P_n$ be the matrix with second unit diagonal.
\smallskip

\begin{consequence}[Bruhat decomposition of matrix $A$]
 
\par\noindent Let matrix $A=(a_{i,j})$  have the rank  $r$, $r\leq n$, and $B=P_nA$. 
Let $B=LDU$ be the LDU-decomposition of matrix $B$. Then
$V=P_nLP_r$ and $U$ are  upper triangular matrices of size $n\times r$ and $r\times n$ correspondingly
and
$$
A=V (P_rD) U
\eqno(5)
$$
is the Bruhat decomposition of matrix $ A$.
\end{consequence}


We are interested in the block form of decomposition algorithms for LDU and Bruhat decompositions.
Let us use some block matrix notations.

For any matrix $A$ (or $A^p_q$) we denote by $A^{i_1,i_2}_{j_1,j_2}$
(or $A^{p;i_1,i_2}_{q;j_1,j_2}$) the block which stands at the intersection of rows  $i_1+1, \ldots ,i_2 $ and columns  $j_1+1, \ldots ,j_2$  of the matrix.  We denote by $A^{i_1 }_{i_2 }$ the diagonal block  $A^{i_1,i_2}_{i_1,i_2}$.


\section{ LDU algorithm }

 \noindent
{\it Input:} ($\cA_n^{k},\al^k$), $0\leq k<n$.

 \noindent
{\it Output:} $\{ L^k_n, \{  \al^{k+1} ,  \al^{k+2}, \ldots , \al^n \},  U^k_n, M^k_n, W^k_n \}$,

 \noindent
 where $D^k_n=\al^k\diag\{\al^k \al^{k+1}, \ldots , \al^{n-1} \al^{n}  \}^{-1}$,
$ M^k_n=\al^k(L^k_n D^k_n)^{-1},\ W^k_n=\al^k(D^k_n U^k_n)^{-1}$.
\bigskip

\noindent
1. If  $k=n-1$,  $\cA_n^{n-1}=(a^n)$ is a matrix of the first order, then we obtain
 $$\{ a^n, \{a^n\}, a^n, a^{n-1}, a^{n-1}  \}, \ \ D^{n-1}_n=(\al^n)^{-1}.$$

\noindent
2. If  $k=n-2$,  $\cA_n^{n-2}=\left(\begin{array}{cc} \al^{n-1} &\be\\\ga&\de \end{array}\right) $ is a matrix of second order, then we obtain

\noindent
$$
\bigg\{
\left(\begin{array}{cc} \al^{n-1}&0\\ \ga& \al^n
\end{array}\right),
 \{ \al^{n-1}, \al^{n} \},
\left(\begin{array}{cc}  \al^{n-1}&\be\\ 0&\al^n
\end{array}\right), 
\left(\begin{array}{cc}\al^{n-2}&0\\- \ga& \al^{n-1}
\end{array}\right),
\left(\begin{array}{cc} \al^{n-2}&-\be\\ 0&\al^{n-1}
\end{array}\right)
\bigg\}
$$

\noindent
where $ \al^n= { (\al^{n-2})}^{-1} \Dmtr { \al^{n-1}} \be  \ga \de $,
  $D^{n-2}_n=\al^{n-2}\diag\{\al^{n-2} \al^{n-1},\al^{n-1} \al^{n}  \}^{-1}$.

\noindent
3. If the order of the matrix $ \cA^{k}_n  $
more than two ($0\leq k <n-2 $), then we choose an integer $ s $ in the interval
$ ( k <s <n) $ and divide the matrix into blocks
$$
\cA_n^{k}=
\left(\begin{array}{cc} \cA_s^{k} &\bB\\ \bC& \bD
\end{array}\right).
\eqno(6)
$$

\noindent
3.1. Recursive step
%
$$
\{L^k_s, \{  \al^{k+1} ,  \al^{k+2}, \ldots , \al^s \},  U^k_s,
 M^k_s,\ W^k_s  \}
= \mathbf{LDU}(\cA_s^{k},\al^k)
$$

\noindent
3.2.  We compute
$$
\widetilde U= (\al^k)^{-1} M^k_s\bB, \ \
   \widetilde L=  (\al^k)^{-1} \bC  W^k_s,
\eqno(7)
$$
$$
\cA_n ^{s}=(\al^k)^{-1}\al^s( \bD -  \widetilde L D^k_s \widetilde U ).
\eqno(8)
$$

\noindent
3.3.   Recursive step
$$
\{L_n ^{s}, \{  \al^{s+1} ,  \al^{s+2}, \ldots , \al^n \},  U_n ^{s},  M^s_n,\ W^s_n  \}
= \mathbf{LDU}(\cA_n ^{s},\al^s )
$$

\noindent
3.4 {\it Result:}
$$
\{L^k_n, \{  \al^{k+1} ,  \al^{k+2}, \ldots , \al^n \}, U^k_n, M^k_n, W^k_n  \},
$$
where
$$
L^k_n=\left(\begin{array}{cc}   L^k_s&0\\ \widetilde L& L^s_n
\end{array}\right),
\
 U^k_n= \left(\begin{array}{cc}   U^k_s& \widetilde U \\ 0& U^s_n
\end{array}\right),
\eqno(9)
$$

$$
M^k_n =\left(\begin{array}{cc}   M^k_s  &0\\
 -   M^s_n \widetilde LD^k_s M^k_s /\al^{k} & \ \  M^s_n \ \
\end{array}\right),
\eqno(10)
$$
$$
 W^k_n = \left(\begin{array}{cc} \ \  W^k_s \ \  &  -  W^k_s D^k_s \widetilde U W^s_n/\al^{k} \\
 0 &   W^s_n
\end{array}\right).
\eqno(11)
$$

\section{Proof of the correctness of the LDU algorithm}

Proof of the correctness of this algorithm is based on several determinant identities.

\begin{definition}[$\delta^k_{i,j}$ minors and  ${\mathcal G}^{k}$ matrices]

Let $A\in R^{n\times n}$ be a matrix.
The determinant of the matrix, obtained from the upper left block
$A^{0,k}_{0,k} $ of matrix $A$
by the replacement in matrix $A$ of the column $i$ by the column $j$ is denoted by
$\delta^k_{i,j}$. The matrix of such minors is denoted by
$$
{\mathcal G}^{k}_s=(\delta^{k+1}_{i,j})
\eqno(12)
$$
\end{definition}

We need the following theorem (see \cite{13} and \cite{14}):

\begin{theorem}[Base minor's identity]
 
Let $A\in R^{n\times n}$ be a matrix
and  $i,j,s,k$, be integers in the intervals: $  0\leq k < s\leq  n, \ 0<i,
 j\leq n$. Then the following identity is true
$$
\al^{s} \al_{ij}^{k+1}-\al^k a_{ij}^{s+1}= \sum_{p=k+1}^{s}
\al_{ip}^{k+1} \delta_{pj}^{s}.
 \eqno(13)
$$
\end{theorem}

The minors $ a_{ij}^{s+1}$ in the left side of this identity
equal zero if $i<s+1$. Therefor this theorem gives the
following

\begin{consequence}  

Let $A\in R^{n\times n}$ be a matrix
and  $s, k$, be integers in the intervals:   $  0\leq k < s \leq  n $.
  Then  the following identities are true
$$
\al^s  U^{k;k+1, s}_{n; s+1, n}= U^{k}_{s} \cG^{k;k+1, s}_{n; s+1, n}.
\eqno(14)
$$
$$
\al^s \cA^{k;k+1, s}_{n; s+1, n}=\cA^{k}_{s} \cG^{k;k+1, s}_{n; s+1, n}. \eqno(15)
$$
\end{consequence}

The block $\cA^{k;k+1, s}_{n; s+1, n}$ of the matrix $\cA^{k}_n$ was denoted by $\bB$.
Due to Sylvester identity we can write the equation for the adjoint matrix
$$(A^k_s)^{*} = (A^k_s)^{-1} (\al^s)(\al^k)^{s-k-1}
\eqno(16)
$$
Let us multiply both sides of equation (15) by adjoint matrix $(A^k_s)^{*}$ and use the equation (16).
Then we get

\begin{consequence}
$$  (A^k_s)^{*} \bB=
  (A^k_s)^{*} \cA^{k;k+1, s}_{n; s+1, n}=  (\al^k)^{s-k-1} \cG^{k;k+1, s}_{n; s+1, n}.
\eqno(17)
$$
\end{consequence}

As well as $L^k_sD^k_sU^k_s=A^k_s$,
$$M^k_s= \al^k(L^k_sD^k_s)^{-1}=\al^k U^k_s (A^k_s)^{-1} \hbox{\  and \ }
W^k_s= \al^k (D^k_s U^k_s)^{-1}.
\eqno(18)
$$
 Therefor
$$\widetilde U= (\al^k)^{-1} M^k_s\bB=(\al^k)^{-1} U^k_s (A^k_s)^{-1}\bB=(\al^s)^{-1}(\al^k)^{-s+k}
U^k_s (A^k_s)^{*}\bB.
\eqno(19)
$$
Equations (19), (17), (14) give the

\begin{consequence}
$$\widetilde U=
 U^{k;k+1, s}_{n; s+1, n}
\eqno(20)
$$
\end{consequence}

In the same way we can prove

\begin{consequence} 
$$
\widetilde L=  L^{k;s+1, n  }_{n; k+1, s}  .
\eqno(21)
$$
\end{consequence}

Now we have to prove the identity (8). Due to the equations (14)-(19) we obtain
$$
\widetilde L D^k_s \widetilde U=(\al^k)^{-1} \bC  W^k_s D^k_s   (\al^k)^{-1} M^k_s\bB=
$$
$$
(\al^k)^{-2} \bC (A^k_s)^{-1} \bB=(\al^k)^{-s+k-1}(\al^s)^{-1} \bC (A^k_s)^{*} \bB
\eqno(22)
$$

The identity
$$\cA_n ^{s}=(\al^k)^{-1}(\al^s \bD - (\al^k)^{-s+k+1} \bC (A^k_s)^{*} \bB)
\eqno(23)
$$
was proved in \cite{13} and \cite{14}.
Due to (20) and (21) we obtain the identity (8).

To prove the formula (10) and (11) it is sufficient to verify the identities 
$ M^k_n=\al^k(L^k_n D^k_n)^{-1}$ and $ W^k_n=\al^k(D^k_n U^k_n)^{-1}$ using (9),(10), (11) and definition  $D^k_n=\al^k\diag\{\al^k \al^{k+1}, \ldots , \al^{n-1} \al^{n}  \}^{-1}$.

\section{Complexity}

\begin{theorem}
The algorithm has the same complexity as matrix multiplication.
\end{theorem}
\begin{proof}

The total amount of matrix multiplications in (7)-(15) is equal to
7 and the total amount of recursive calls is equal to 2. We do not
consider multiplications of the diagonal matrices.

We can compute the decomposition  of the second order matrix  by
means of 7 multiplicative operations. Therefore  we get the
following recurrent equality for complexity
$$
t(n)=2 t(n/2)+  7 M(n/2),\ \  t(2)=7.
$$
Let $\gamma$ and $\beta$ be constants, $3\geq\beta>2$, and let
  $M(n)= \gamma n^{\beta} + o(n^{\beta})$ be the number of multiplication
  operations in one $n\times n$ matrix multiplication.

After summation from  $n=2^k$ to $2^1$ we obtain
$$
 7  \gamma(2^0 2^{\beta\cdot  (k-1)} + \ldots + 2^{k-2}2^{\beta \cdot 1})+2^{k-2}
7 =  7 \gamma\frac{n^{\beta}- n 2^{\beta-1} }{2^{\beta}-2} +
\frac{7}{4}n.
$$
Therefore the complexity of the decomposition is
$$
\sim
\frac{  7 \gamma n^{\beta}}{2^{\beta}-2}
$$
\end{proof}
 
\section{The exact triangular decomposition}

\begin{definition}
A decomposition of the matrix $ A $ of rank $ r $ over a commutative domain $ R $ in the product of five matrices
$$
A = PLDUQ
\eqno (24)
$$
is called {\em exact triangular decomposition} if $ P $ and $ Q $ are permutation matrces,
$ L $ and $ PLP^T $ are nonsingular lower triangular matrices, 
$ U $ and $ Q^T UQ $ are nonsingular upper triangular matrices over $ R $, 
$ D = \diag (d_1^{-1}, d_2^{-1 }, .., d_r^{-1}, 0, .., 0) $ is a diagonal matrix of rank $ r $,  $d_i\in R \backslash \{0\}$, $i=1,..r$. 
\end{definition}

Designation: ${\cal \mathbf  {ETD}}(A)=(P,L,D,U,Q)$.

\begin{theorem}[Main theorem] Any matrix over a commutative domain has an exact triangular decomposition.
\end{theorem}

Before proceeding to the proof, we note that the exact triangular decomposition relates the $ LU $ decomposition and the Bruhat decomposition in the field of fractions.

If $ D $ matrix is combined with $ L $ or $ U $, we get the expression $ A = PLUQ $. This is the $LU$-decomposition with permutations of rows and columns. If the factors are grouped in the following way:
$$
A=(PLP^T)(PDQ)(Q^TUQ), 
$$
then we obtain $ \bL \bD \bU $-decomposition.  If $S$ is a permutation matrix in which the unit elements are placed on the secondary diagonal, then $ (S \bL S) (S^T \bD) \bU $ is the Bruhat decomposition of the matrix $(SA)$. 

Bruhat decomposition can be obtained from those $ PLUQ $-decomposition that satisfy the additional conditions:
  matrix $ PLP^T $ and $ Q^TUQ $ are triangular. Conversely, $LU$-decomposition can be obtained from the Bruhat decomposition $ V'D'U'$. This can be done if the permutation matrix $ D $ can be decomposed into a product of permutation matrices $ D'= PQ $ so that the $ P^TL'P $ and  $ QU'Q^T $ are triangular matrices.

If matrix $ A $ is a zero matrix, then ${\cal \mathbf  {ETD}}(A)=(I,I,0,I,I)$. 

If $ A $ is a nonzero matrix of the first order, then ${\cal \mathbf  {ETD}}(A)=(I,a,a^{-1},a,I)$. 

Let us consider a non-zero matrix of order two. We denote
$$
\cA =\left(\begin{array}{cc} \al  &\be\\\ga&\de \end{array}\right), \ 
\Delta =\left|\begin{array}{cc} \al  &\be\\\ga&\de \end{array}\right|,
\ve= \left\{\begin{array}{l } \Delta, \ \Delta\neq 0  \\
1, \ \Delta= 0.  \end{array}\right.
$$
Depending on the location of zero elements, we consider four possible cases. For each case, we give the exact triangular decomposition: \\
$
\hbox{If }\  \al\neq 0,\ \hbox{ then }\ 
\cA =\left(\begin{array}{cc} \al  &0\\\ga&\ve \end{array}\right)
\left(\begin{array}{cc} \al^{-1}  &0\\0&\De^{-1} \al^{-1}\end{array}\right) 
\left(\begin{array}{cc} \al  &\be\\0&\ve \end{array}\right). 
$ \\
$
\hbox{If }\ \al=0, \  \be\neq 0,\ \hbox{ then }\ 
\cA =\left(\begin{array}{cc} \be  &0\\ \de&\ve \end{array}\right)
\left(\begin{array}{cc} \be^{-1}  &0\\0& -\De^{-1} \be^{-1}\end{array}\right) 
\left(\begin{array}{cc} \be  &0 \\ 0&\ve \end{array}\right)
\left(\begin{array}{cc} 0  &1 \\ 1&0 \end{array}\right).
$ \\
$
\hbox{If }\ \al=0, \  \ga\neq 0,\ \hbox{ then }\ 
\cA =\left(\begin{array}{cc} 0  &1 \\ 1&0 \end{array}\right)
\left(\begin{array}{cc} \ga  &0\\ 0&\ve \end{array}\right)
\left(\begin{array}{cc} \ga^{-1}  &0\\0& -\De^{-1} \ga^{-1}\end{array}\right) 
\left(\begin{array}{cc} \ga  &\de \\ 0&\ve \end{array}\right).  
$ \\
$
\hbox{If }\ \al=\be=\ga=0,  \de\neq 0,  \hbox{ then }\ 
\cA =\left(\begin{array}{cc} 0  &1 \\ 1&0 \end{array}\right)
\left(\begin{array}{cc} \de  &0\\ 0&1 \end{array}\right)
\left(\begin{array}{cc} \de^{-1}  &0\\0& 0\end{array}\right) 
\left(\begin{array}{cc} \de  &0 \\ 0&1 \end{array}\right)
\left(\begin{array}{cc} 0  &1 \\ 1&0 \end{array}\right).
$ \\

There are only two different cases  for matrices of size $ 1 \times 2$: \\
$
\hbox{If }\ \al \neq 0,\ \hbox{ then }\ 
\left(\begin{array}{cc} \al & \be \end{array}\right) = 
\left(\begin{array}{cc} \al  \end{array}\right)
\left(\begin{array}{cc} \al^{-1} & 0\end{array}\right) 
\left(\begin{array}{cc} \al & \be \\ 0& 1 \end{array}\right).
$ \\
$
\hbox{If }\ \al=0, \  \be \neq 0,\ \hbox{ then }\ 
\left(\begin{array}{cc} 0 & \be \end{array}\right) = 
\left(\begin{array}{cc} \be  \end{array}\right)
\left(\begin{array}{cc} \be^{-1} & 0\end{array}\right) 
\left(\begin{array}{cc} \be & 0 \\ 0& 1 \end{array}\right)
\left(\begin{array}{cc} 0  &1 \\ 1&0 \end{array}\right).
$ \\
Two   cases for matrices of size $ 2 \times 1$ can be easily obtained by a simple transposition.

These examples allow us to formulate

\begin{sentence}

For all matrices $ \cA $ of size $ n \times m $, $ n, m <3 $ there exists an exact triangular decomposition.
\end{sentence}

In addition, we can formulate the following property, which holds for triangular matrices and permutation matrices in the exact triangular decomposition.

We denote by $ I_s $ the identity matrix of order $ s $.

\begin{property}[Property of the factors]
For a matrix $ A \in R^{n \times m} $ of rank $ r $, $ r <n, r <m $ over a commutative domain $ R $ there exists the exact triangular decomposition (24) in which 

\noindent
($\alpha$) the matrices $ L $ and $ U $ are of the form
$$
L = \left(\begin{array}{cc} L_1  &0\\ L_2 & I_{n-r} \end{array}\right) 
U= \left(\begin{array}{cc} U_1  &U_2\\0&I_{m-r} \end{array}\right),
\eqno(25)
$$
($\beta$)  the matrices $ PLP^T $ and $ Q^TUQ $ remain triangular after replacing in the matrices $L$ and $Q$ of unit blocks $ I_{n-r} $ and $ I_{m-r} $ by arbitrary triangular blocks.
\end{property}

Without loss of generality of the main theorem, we shall prove it for the exact triangular decompositions with the property 1.
We prove it by induction. The theorem is true for matrices of sizes smaller than three. 

We consider a matrix $ \cA $ of size $ N \times M $. Assume that all matrices   of size less than $ n \times m $ have the exact triangular decomposition. We split the matrix $ \cA $ into blocks:
$
\cA =  \left(\begin{array}{cc} \bA  &\bB\\ \bC & \bD \end{array}\right),
$
where   $\bA\in R^{n\times n}$, $n<N$, $n<M$.
 
(1). Let the block $\bA$ have the full rank. There exists exact triangular decomposition of this block: 
$  
\bA=P_1L_1D_1U_1Q_1.  
$ Here the diagonal matrix $ D_1 $ has full rank and the matrix $\cA $ is decomposed into the factors:
$$
\left(\begin{array}{cc} P_1  &0\\ 0 & I \end{array}\right)
\left(\begin{array}{cc} L_1  &0\\  \bC Q_1^{T}U_1^{-1}D_1^{-1} & I \end{array}\right)
\left(\begin{array}{cc} D_1  &0\\ 0 & \bD  \end{array}\right)
\left(\begin{array}{cc} U_1  &D_1^{-1}L_1^{-1}P_1^{T}\bB\\ 0 & I \end{array}\right)
\left(\begin{array}{cc} Q_1  &0\\ 0 & I \end{array}\right). 
$$
The matrix $\bD $
also has the exact triangular decomposition
$\bD =P_2L_2D_2U_2Q_2$. Substituting it in this decomposition, we obtain a new decomposition of the matrix $\cA$: 
$$
\left(\begin{array}{cc} P_1  &0\\ 0 & P_2 \end{array}\right)
\left(\begin{array}{cc} L_1  &0\\ P_2^T\bC Q_1^{T}U_1^{-1}D_1^{-1} & L_2 \end{array}\right)
\left(\begin{array}{cc} D_1  & 0 \\ 0 & D_2 \end{array} \right)
\left(\begin{array}{cc} U_1  &D_1^{-1}L_1^{-1}P_1^{T}\bB Q_2^{T}\\ 0 & U_2 \end{array}\right)
\left(\begin{array}{cc} Q_1  &0\\ 0 & Q_2 \end{array}\right). 
$$
 
It is easy to see that this decomposition is exact triangular if both block decompositions were exact triangular.

(2)   
Let the block $\bA$ has rank $r$, $r<n$. There exists exact triangular decomposition of this block: 

$$  
 \bA  =P_1L_1D_1U_1Q_1.  
$$
Here $U_1=\left(\begin{array}{cc} U_0  &V_0\\ 0 & I  \end{array}\right)$, 
$L_1=\left(\begin{array}{cc} L_0  & 0\\ M_0 & I  \end{array}\right)$ and
the diagonal matrix $D_1=\left(\begin{array}{cc} d_1  &0\\ 0 & 0  \end{array}\right)$ has a block 
$d_1$ of rank $r$.

Let us denote  $(\bC_0, \bC_1)$= 
$ \bC Q_1^{T} \left(\begin{array}{cc} U_0^{-1}  & -V_0\\ 0 & I  \end{array}\right) $ and 
$\left(\begin{array}{c} \bB_0  \\ \bB_1 \end{array}\right)$= 
$\left(\begin{array}{cc} L_0^{-1}  & 0\\ -M_0 & I  \end{array}\right) P_1^{T}\bB$.
Then for the matrix $\cA$ we obtain the decomposition:
$$ \cA=
\left(\begin{array}{cc} P_1  &0\\ 0 & I \end{array}\right)
\left(\begin{array}{ccc}  L_0  & 0 &0\\ M_0 & I  &0\\  \bC_0 d_1^{-1} &0& I \end{array}\right)
\left(\begin{array}{ccc}  d_1 &0 &0\\ 0 & 0 & \bB_1 \\0 & \bC_1&\bD \end{array} \right)
\left(\begin{array}{ccc}  U_0 & V_0 &d_1^{-1}\bB_0\\ 0 & I & 0 \\ 0 & 0 & I  \end{array}\right)
\left(\begin{array}{cc} Q_1  &0\\ 0 & I \end{array}\right). 
\eqno(26)
$$
(2.1) Let $\bB_1=0$ and $ \bC_1=0$. We can rearrange the  block $ \bD $ in the upper left corner
$$
 \left(\begin{array}{cc}   0 & \bB_1 \\ \bC_1&\bD \end{array} \right)=
 \left(\begin{array}{cc}   0 & I \\ I & 0 \end{array} \right) 
 \left(\begin{array}{cc}   \bD & 0 \\ 0 &  0 \end{array} \right)
 \left(\begin{array}{cc}  0 & I \\ I & 0 \end{array} \right). 
$$
 Let us find the exact triangular decomposition of $\bD$:
$$
 \bD  =P_2L_2D_2U_2Q_2.
$$ 
We denote   
$$
\bP_3=
\left(\begin{array}{cc} P_1  &0\\ 0 & P_2 \end{array}\right)
\left(\begin{array}{ccc}  I  & 0 &0\\ 0 & 0 &I\\  0 &I&0 \end{array}\right)
, \ 
\bQ_3=
\left(\begin{array}{ccc}  I  & 0 &0\\ 0 & 0 &I\\  0 &I&0 \end{array}\right)
\left(\begin{array}{cc} Q_1  &0\\ 0 & Q_2 \end{array}\right).
$$
Then for the matrix $\cA$ we obtain the following decomposition:
$$  \cA=
\bP_3
\left(\begin{array}{ccc}  L_0  & 0 &0\\ P_2^T \bC_0 d_1^{-1} & L_2  &0\\  M_0 &0& I \end{array}\right)
\left(\begin{array}{ccc }  d_1 & 0 &0 \\ 0 & D_2 &0 \\ 0 & 0 & 0  \end{array} \right)
\left(\begin{array}{ccc}  U_0 & d_1^{-1}\bB_0Q_2^T & V_0\\ 0 & U_2 &0 \\ 0 & 0 & I  \end{array}\right)
\bQ_3
. 
$$
It is easy to check that the decomposition is exact triangular.

(2.2) Suppose that at least one of the two blocks of $ \bB_1 $ or $ \bC_1 $ is not zero. Let the exact triangular decomposition exist for these blocks:
$$
 \bC  =P_2L_2D_2U_2Q_2, \  \bB  =P_3L_3D_3U_3Q_3.
$$
  We denote 
$$
\bP_1= \left(\begin{array}{cc} P_1  &0\\ 0 & I \end{array}\right),
\bP_2=
\left(\begin{array}{ccc}  I & 0 & 0\\ 0 & P_3 & 0\\  0 & 0 &P_2 \end{array}\right)
, \ 
\bQ_2= \left(\begin{array}{ccc}  I  & 0 &0\\ 0 & Q_2 &0\\  0 &0&Q_3 \end{array}\right), \
\bQ_1=
\left(\begin{array}{cc} Q_1  &0\\ 0 & I \end{array}\right),\
$$
$\bP_3=\bP_1\bP_2$,
$\bQ_3=\bQ_2\bQ_1$,
$\bD'=   L_2^{-1}P_2^T\bD Q_3^T U_3^{-1}$.

Then, basing on the expansion (26) we obtain for the matrix $ \cA $ the  decomposition of the form:
$$
\cA= \bP_3
\left(\begin{array}{ccc}  L_0  & 0 &0\\P_3^T M_0 & L_3  &0\\ P_2^T \bC_0 d_1^{-1} &0& L_2 \end{array}\right)
\left(\begin{array}{ccc }  d_1 & 0 &0\\ 0 & 0  &D_3 \\ 0 & D_2  &\bD' \end{array} \right)
\left(\begin{array}{ccc}  U_0 & V_0 Q_2^T &d_1^{-1}\bB_0 Q_3^T\\ 0 & U_2 & 0 \\ 0 & 0 & U_3 \end{array}\right)
\bQ_3. 
\eqno(27)
$$
We denote $ d_2 $ and $ d_3 $ nondegenerate blocks of the matrices $ D_2 $ and $ D_3 $, respectively,  
$$
(V_1,V_4)= V_0 Q_2^T,    (V_5,V_6)= d_1^{-1}\bB_0 Q_3^T ,  
 \left(\begin{array}{c} M_1 \\ M_4\end{array} \right) =  P_3^T M_0,   
 \left(\begin{array}{c} M_5\\ M_6\end{array} \right) =  P_2^T \bC_0 d_1^{-1}
$$ $$
 L_2=
\left(\begin{array}{cc} L_2'  &0\\ M_2 & I \end{array}\right),
 L_3=
\left(\begin{array}{cc} L_3'  &0\\ M_3 & I \end{array}\right),
 U_2=
\left(\begin{array}{cc} U_2'  &V_2\\ 0 & I \end{array}\right),
 U_3=
\left(\begin{array}{cc} U_3'  &V_3\\ 0 & I \end{array}\right),
 \bD'=
\left(\begin{array}{cc} \bD'_1 & \bD'_3 \\  \bD'_2  & \bD'_4\end{array} \right).
$$
$$
M_7= \bD'_2 d_3^{-1}, \ 
V_7=d_2^{-1}\bD'_1 U_3', \ 
V_8=d_2^{-1}(\bD'_1 V_3+\bD'_3).
$$
Then (27) can be written as
 $$
\cA= \bP_3
\left(\begin{array}{ccccc } L_0 & 0 & 0 & 0 & 0\\ M_1  & L_3'  & 0 & 0 & 0 
\\ M_4 & M_3 & I & 0 & 0 \\M_5& 0 & 0  & L_2' & 0 \\M_6 & 0 &0 & M_2 & I\end{array} \right)
\left(\begin{array}{ccccc }  d_1 & 0 & 0 & 0 & 0\\ 0 & 0  & 0 & d_3 & 0 
\\ 0 & 0 & 0 & 0 & 0 \\0& d_2 & 0  & \bD'_1 & \bD'_3 \\ 0 & 0 &  0 &\bD'_2&\bD'_4\end{array} \right)
\left(\begin{array}{ccccc } U_0 & V_1 & V_4 & V_5 & V_6\\ 0 & U_2'  & V_2 & 0 & 0 
\\ 0 & 0 & I & 0 & 0 \\0& 0 & 0  & U_3' & V_3 \\0 & 0 &0 & 0 & I\end{array} \right)
\bQ_3=
$$ 
$$ \bP_3
\left(\begin{array}{ccccc } L_0 & 0 & 0 & 0 & 0\\ M_1  & L_3'  & 0 & 0 & 0 
\\ M_4 & M_3 & I & 0 & 0 \\M_5& 0 & 0  & L_2' & 0 \\M_6 & M_7 &0 & M_2 & I\end{array} \right)
\left(\begin{array}{ccccc }  d_1 & 0 & 0 & 0 & 0\\ 0 & 0  & 0 & d_3 & 0 
\\ 0 & 0 & 0 & 0 & 0 \\0& d_2 & 0  & 0 & 0 \\ 0 & 0 &  0 & 0 & \bD'_4 \end{array} \right)
\left(\begin{array}{ccccc } U_0 & V_1 & V_4 & V_5 & V_6\\ 0 & U_2'  & V_2 & V_7 & V_8 
\\ 0 & 0 & I & 0 & 0 \\0& 0 & 0  & U_3' & V_3 \\0 & 0 &0 & 0 & I\end{array} \right)
\bQ_3.
\eqno(28)
$$
Find the exact triangular decomposition  $\bD'_4$:
$$
 \bD'_4  =P_4L_4D_4U_4Q_4, 
\eqno(29)
$$
Let us denote  the matrices $\bP_4=\diag(I,I,I,I,P_4)$, $\bQ_4=\diag(I,I,I,I,Q_4)$, $\bP_5=\bP_3 \bP_4$, $\bQ_5=\bQ_4 \bQ_3$, 
$(M_6', M_7', M_2') =P_4^T(M_6, M_7, M_2)$ и $(V_6', V_8', V_3')=(V_6, V_8, V_3)Q_4^T$.

After substituting (29) into (28) we obtain the decomposition of the matrix $\cA $ as
$$ \cA= \bP_5
\left(\begin{array}{ccccc } L_0 & 0 & 0 & 0 & 0\\ M_1  & L_3'  & 0 & 0 & 0 
\\ M_4 & M_3 & I & 0 & 0 \\M_5& 0 & 0  & L_2' & 0 \\M_6' & M_7' &0 & M_2' & L_4\end{array} \right)
\left(\begin{array}{ccccc }  d_1 & 0 & 0 & 0 & 0\\ 0 & 0  & 0 & d_3 & 0 
\\ 0 & 0 & 0 & 0 & 0 \\0& d_2 & 0  & 0 & 0 \\ 0 & 0 &  0 & 0 &  D_4 \end{array} \right)
\left(\begin{array}{ccccc } U_0 & V_1 & V_4 & V_5 & V_6'\\ 0 & U_2'  & V_2 & V_7 & V_8' 
\\ 0 & 0 & I & 0 & 0 \\0& 0 & 0  & U_3' & V_3' \\0 & 0 &0 & 0 & U_4\end{array} \right)
\bQ_5.
\eqno(30)
$$

We rearrange the blocks $ d_2 $, $ d_3 $ and $ D_4 $ to obtain the diagonal matrix 
${\mathbf d}=\diag(d_1,d_3,d_2,D_4,0)$. To do it we use permutation matrices $ P_6 $ and $ Q_6 $:
$$
P_6=
\left(\begin{array}{ccccc }   1 & 0 & 0 & 0 & 0\\ 0 & 0  & 0 & 1 & 0 
\\ 0 & 1 & 0 & 0 & 0 \\0 & 0 &0 & 0 & 1 \\0& 0 & 1  & 0 & 0  \end{array} \right),
Q_6=
\left(\begin{array}{ccccc } 1 & 0 & 0 & 0 & 0\\ 0 & 1  & 0 & 0 & 0 
\\ 0 & 0 & 0 & 0 & 1 \\0& 0 & 1  & 0&0 \\0& 0 &0 &1&0\end{array} \right),
P_6
\left(\begin{array}{ccccc }  d_1 & 0 & 0 & 0 & 0\\ 0 & 0  & 0 & d_3 & 0 
\\ 0 & 0 & 0 & 0 & 0 \\0& d_2 & 0  & 0 & 0 \\ 0 & 0 &  0 & 0 & D_4 \end{array} \right)
Q_6={\mathbf d}
.
$$
As a result, we  obtain the decomposition:
$$
\cA= \bP_6 \bL {\mathbf d} \bU \bQ_6, 
\eqno(31)
$$
with permutation matrices $\bP_6=\bP_5 P_6^T$ and $\bQ_6= Q_6^T \bQ_5$, 
diagonal matrix ${\mathbf d}$ and triangular matrices
$$
\bL=P_6 
\left(\begin{array}{ccccc } L_0 & 0 & 0 & 0 & 0\\ M_1  & L_3'  & 0 & 0 & 0 
\\ M_4 & M_3 & I & 0 & 0 \\M_5& 0 & 0  & L_2' & 0 \\M_6' & M_7' &0 & M_2' & L_4\end{array} \right)
P_6^T=
\left(\begin{array}{ccccc } L_0 & 0 & 0 & 0 & 0\\ M_5 & L_2'  & 0 & 0 & 0 
\\ M_1  & 0 & L_3' & 0 & 0 \\M_6'& M_7' & M_2' & L_4 & 0 \\M_4 & 0 &M_3 & 0 & I\end{array} \right)
$$
$$
\bU=Q_6^T 
\left(\begin{array}{ccccc }
 U_0 & V_1 & V_4 & V_5 & V_6'
\\ 0 & U_2'  & V_2 & V_7 & V_8' 
\\ 0 & 0 & I & 0 & 0 
\\0& 0 & 0  & U_3' & V_3' 
\\0 & 0 &0 & 0 & U_4\end{array} \right)
Q_6=
\left(\begin{array}{ccccc }
 U_0 & V_1& V_5 & V_6'& V_4  
\\ 0 & U_2'& V_7 & V_8' & V_2    
\\0& 0 &  U_3' & V_3' &0  
\\0 & 0 & 0 & U_4 & 0
\\ 0 & 0 & 0 & 0 & I\end{array} \right)
$$
We show that the expansion (31) is an exact triangular decomposition. To do this, we must verify that the matrices $ {\cal L}=\bP_6 \bL \bP_6^{T}$
and ${\cal Q}=\bQ_6^{T} \bU \bQ_6$ are triangular, and the matrices  $\bP, \bL,\bU, \bQ$ satisfy the properties ($ \alpha $) and ($ \beta $).

It is easy to see that all matrices in sequence
$${\cal L}_1 =  P_6 \bL P_6^{T}, 
{\cal L}_2=\bP_4{\cal L}_1 \bP_4^T, 
{\cal L}_3=\bP_2{\cal L}_2\bP_2^T,
{\cal L}_4=\bP_1{\cal L}_3\bP_1^T
\eqno(32)
$$
are triangular and
${\cal L}_4= {\cal L}$. 

Similarly, all of the matrices in the sequence
$${\cal U}_1 =  Q_6^T \bL Q_6, 
 {\cal U}_2=\bQ_4^T{\cal U}_1 \bQ_4, 
 {\cal U}_3=\bQ_2^T{\cal U}_2\bQ_2,
 {\cal U}_4=\bQ_1^T{\cal U}_3\bQ_1 
\eqno(33)
$$ 
are triangular and 
${\cal U}_4= {\cal U}.$ 

For the matrices $ \bL $ and $ \bU $ Property 1 ($ \alpha $) is satisfied.
To verify the properties ($ \beta $),
the  unit block in the lower right corner of the matrix $ \bL $ and $ \bU $ should be replaced by an arbitrary triangular block,  respectively, the lower triangle for $ \bL $ and the upper triangular for $ \bU $. We check that all the matrices in (32) and (33) will be still triangular. This is based on the fact that the exact triangular decompositions for matrices $ \bA, \bB, \bC, \bD '$ have the property ($\beta $).
 
\section{Conclusion}

Algorithms for finding the LDU and Bruhat decomposition
in commutative domain are described. These algorithms have the same
complexity as matrix multiplication.

\section{Example}
  $ \left[\begin{array}{cccc}1 & -4&0&1\\ 4& 5&5&3\\ 1&2&2&2\\ 3&0&0&1\\ \end{array}\right]$
  $  = $
$  \begin{array}{c}
\left[\begin{array}{cccc}-24& 0&  12& 1\\
 0&   60& 15& 4\\
 0&   0&  6&  1\\
 0&   0&  0&  3\\
\end{array}\right]
\left[\begin{array}{cccc}0&     0&      1/(-144)& 0          \\
 0&     0&      0&          1/(-1440)\\
 0&     1/18& 0&          0          \\
 1/3& 0&      0&          0          \\
\end{array}\right]
\left[\begin{array}{cccc}3& 0& 0&   1  \\
 0& 6& 6&   5  \\
 0& 0& -24& -16\\
 0& 0& 0&   60 \\
\end{array}\right]
 \end{array} $




\begin{thebibliography}{99}
%
\bibitem{10}
Grigoriev D. Analogy of Bruhat decomposition for the closure of a cone of Chevalley group of a classical serie.
Soviet Math. Dokl., vol.23, N 2, 393-397 (1981).
%
\bibitem{11}
Grigoriev D.
Additive complexity in directed computations. Theoretical Computer Science,  vol.19, 39-67 (1982).
%
\bibitem{12}
Malaschonok G.I.: Fast Generalized Bruhat Decomposition. In: Ganzha, V.M., Mayr, E.W., Vorozhtsov, E.V. (eds.) 12th
International Workshop on Computer Algebra  in Scientific Computing (CASC 2010),  194-202.
LNCS 6244. Springer, Berlin Heidelberg, (2010).
%
\bibitem{13}
Malaschonok G.I.: Matrix computational methods in commutative rings.  Monograph. Tambov,
Tambov University Publishing House (2002).
%
\bibitem{14}
  Malaschonok G.I.: Effective Matrix Methods in Commutative Domains.  Formal
Power Series and Algebraic Combinatorics.  pp.~506-517. Springer, Berlin  (2000).
%
\bibitem{15}
 Malaschonok G.I.: A Fast Algorithm for Adjoint Matrix Computation, Tambov University Reports, V.5, no. 1,  142-146  (2000).

\bibitem{16}
  Malaschonok G.I.: Fast matrix decomposition in parallel computer algebra. Tambov University Reports, V.15, no.4,  1372-1385 (2010).
 
\bibitem{17}
 Malaschonok G.I.: On the fast generalized Bruhat decomposition in domains, 
 Tambov University Reports, vol. 17, no. 2, 544-550 (2012). 

\end{thebibliography}
\end{document}